\def\OPTIONConf{1}

\ifnum\OPTIONConf=1%
\documentclass[noderivs]{eptcs} %
\else
\documentclass[10pt,leqno]{article}
\fi
\usepackage[overload]{textcase}
\ifnum\OPTIONConf=1%
    \newcommand{\keywordfontsize}{11pt}  %
\else%
   \usepackage[headings]{fullpage}
   \newcommand{\keywordfontsize}{11pt}  %
\fi
\usepackage[bf,hang,ruled]{caption}
\usepackage{color}
\usepackage{breakurl}
\usepackage{pstricks,pst-node,graphicx,eepic}

\usepackage{listings}
\usepackage[authoryear]{natbib}
\bibpunct{(}{)}{;}{a}{}{,}

\usepackage{pgf,tikz}
\usetikzlibrary{shapes,arrows,matrix,fit,backgrounds,calc,decorations,decorations.pathmorphing}

\usepackage{xspace}
\usepackage{euler}
\renewcommand{\ttdefault}{cmtt}   %

\usepackage{amsmath,amssymb,amsthm}

\usepackage{bm}

\usepackage{multirow}
\usepackage{stmaryrd}
\usepackage{pifont}
\usepackage{enumerate}
\usepackage{comment}
\usepackage{mathpartir} \usepackage{proof}
\def\MathparLineskip{\lineskiplimit=0.8em\lineskip=0.7em plus 0.3em}
\usepackage{framed}
\usepackage{rotating}

\let\MathRightArrow\Rightarrow %
\usepackage{marvosym} %
\def\Rightarrow{\MathRightArrow}

\newtheorem{theorem}{Theorem}

\newtheorem{lemma}[theorem]{Lemma}
\newtheorem{corollary}[theorem]{Corollary}
\newtheorem*{corollary*}{Corollary}

\newtheorem*{conjecture*}{Conjecture}
\newtheorem{exercise}[theorem]{Exercise}
\newtheorem{example}[theorem]{Example}
\newtheorem*{example*}{Example}

\newtheorem*{theorem*}{Theorem}
\theoremstyle{remark} 
\theoremstyle{definition} \newtheorem{definition}[theorem]{Definition}

\usepackage{srcltx}

\newcommand{\xLam}{\lambda}
\newcommand{\Lam}[1]{\xLam#1.\,}

\renewcommand{\bfdefault}{b}             

\newcommand{\arr}{\rightarrow}
\newcommand{\entails}{\,\vdash\,}         

\newcommand{\subtype}{\mathrel{\leq}}       

\newcommand{\sectty}{\mathrel{\land}}          
\newcommand{\unty}{\mathrel{\lor}}          

\newcommand{\UnivSym}{\Pi}
\newcommand{\ExisSym}{\Sigma}
\newcommand{\Univ}[2]{\UnivSym{#1}{:}{#2}.\,}
\newcommand{\Exis}[2]{\ExisSym{#1}{:}{#2}.\,}

\newcommand{\indexeq}{\doteq}           

\newcommand{\step}{\mapsto}

\newcommand{\Let}[3]{{\textkw{let}\:#1\,{\texttt{=}}\,#2\:\textkw{in}\:#3}}

\newcommand{\tyname}[1]{\ensuremath{\mathsf{#1}}}
\newcommand{\keyword}[1]{\text{\usefont{T1}{\rmdefault}{b}{n}\fontsize{\keywordfontsize}{11pt}\selectfont#1}}
\newcommand{\textkw}[1]{\keyword{#1}}

\newcommand{\Bits}{\tyname{bits}}

\newcommand{\List}{\tyname{list}}
\newcommand{\Even}{\tyname{even}}
\newcommand{\Odd}{\tyname{odd}}

\newcommand{\arrayenvb}[1]{\renewcommand{\arraystretch}{1}  \begin{array}[b]{@{}c@{}}#1\end{array}}

\newcommand{\Unit}{\tyname{unit}}

\newcommand{\unit}{\texttt{()}}

\newcommand{\against}{\Leftarrow}
\newcommand{\has}{\Rightarrow}

\newcommand{\rulesep}{~~~~~~~}

\newcommand{\Dee}{\mathcal{D}}

\newcounter{codeLineCntr}

\newcommand{\out}[1]{}

\renewcommand{\phi}{\varphi}

\newcommand{\Figureref}[1]{Figure \ref{#1}}

\newcommand{\Sectionref}[1]{Section \ref{#1}}

\newcommand{\val}{~\mathsf{value}}

\newcommand{\intIS}{\mathsf{int}}

\newcommand{\Lbrack}{\char"5B}
\newcommand{\Rbrack}{\char"5D}

\definecolor{lred}{rgb}{1.0, 0.3, 0.3}
\newrgbcolor{lred}{1.0 0.3 0.3}

\newcommand{\BLACKNODE}[1]{~~\raise4pt\hbox{\psellipse[fillstyle=solid,fillcolor=black](0, 0)(6pt,6pt)}\hspace{-3pt}{\textcolor{white}{\textbf{#1}}}~\:}
\newcommand{\REDNODE}[1]{~~\raise4pt\hbox{\psellipse[fillstyle=solid,fillcolor=lred](0, 0)(6pt,6pt)}\hspace{-3pt}{\textbf{#1}}~\:}

\newlength\zzskipwidthlen

\newcommand{\bnfas}{\mathrel{::=}}
\newcommand{\bnfalt}{\mathrel{\mid}}

\newcommand{\doublecomma}{{\,,\hspace{-0.2em},\hspace{0.2em}}}
\newcommand{\Merge}[2]{{#1} \doublecomma {#2}}

\newcommand{\Rmerge}[1]{\ensuremath{\textsf{merge$\against$}_{#1}}\xspace}

\newcommand{\Rmergeup}[1]{\ensuremath{\textsf{merge$\has$}_{#1}}\xspace}
\newcommand{\Rmergeupx}{\ensuremath{\textsf{merge$\has$}}\xspace}

\ifnum\OPTIONConf=1%
\def\OPTIONLoudLabels{0}
\else
\def\OPTIONLoudLabels{1}
\fi
\newcommand{\textt}[1]{\texttt{1}}

\newdimen{\zzzpbox}

\newdimen\zzfontsz
\newcommand{\fontsz}[2]{\zzfontsz=#1%
{\fontsize{\zzfontsz}{1.2\zzfontsz}\selectfont{#2}}}

\newlength{\zzsplatboldwidth}
\newcommand{\xsplatbold}[2]{\settowidth{\zzsplatboldwidth}{{#2}}{#2}\addtolength{\zzsplatboldwidth}{-#1}\hspace{-\zzsplatboldwidth}\raisebox{#1}{{#2}}}
\newcommand{\splatbold}[1]{\xsplatbold{-0.04mm}{#1}}

\newlength{\zzsooperboldwidth}

\makeatletter
\def\url@MGstyle{%
\def\UrlFont{\tiny\ttfamily}
\def\do@url@hyp{\do\-}
\Url@do
}
\makeatother

\ifnum\OPTIONLoudLabels=1%
\newcommand{\Label}[1]{\LoudLabel{#1}}%
\newcommand{\FLabel}[1]{\label{#1}%
{\tt\scriptsize{#1}}}%
\else%
\newcommand{\Label}[1]{\label{#1}}%
\newcommand{\FLabel}[1]{\label{#1}}%
\fi

\def \TirNameStyle #1{{#1}}
\newcommand{\Infer}[3]{\inferrule*[right={\text{\strut#1}}]{{}#2\mathstrut}{{}#3\mathstrut}}

\newcommand{\BeginProof}{\renewcommand{\arraystretch}{1.1} \begin{tabular}[b]{r@{}r @{} l  l}}
\newcommand{\EndProof}{\end{tabular} \renewcommand{\arraystretch}{\mydefaultarraystretch}}

\newenvironment{llproof}{\BeginProof}{\EndProof}

\newcommand{\proofheading}[1]{}  

\newcommand{\erasetypes}[1]{\ensuremath{|#1|}}

\newcommand{\ctxsubtype}{\mathrel{\,\lesssim\,}}

\newcommand{\Rctxanno}{\ensuremath{\mathsf{ctx}\textsf{-}\mathsf{anno}}\xspace}

\newcommand{\RctxSempty}{\ensuremath{{\lesssim}\mathsf{\textsf{-}empty}}\xspace}
\newcommand{\RctxSivar}{\ensuremath{{\lesssim}\mathsf{\textsf{-}ivar}}\xspace}
\newcommand{\RctxSpvar}{\ensuremath{{\lesssim}\mathsf{\textsf{-}pvar}}\xspace}

\newgray{grayboxgray}{0.85}
\newcommand{\textgraybox}[1]{\psframebox[framesep=0pt,fillcolor=grayboxgray,linewidth=0.5pt]{\fbox{\parbox[s][\totalheight]{0mm}{}#1}}}

\ifnum\OPTIONConf=1
  \newcommand{\judgboxfontsize}[1]{\large #1}
\else
  \newcommand{\judgboxfontsize}[1]{\fontsz{10pt}{#1}}
\fi

\newcommand{\judgbox}[2]{%
      {\raggedright \textgraybox{\judgboxfontsize{\ensuremath{#1}}}\!\begin{tabular}[c]{l} #2 \end{tabular} %
      \ifnum\OPTIONConf=1 \\[0pt] \else \\[6pt] \fi%
}}
\newcommand{\judgboxtwelf}[3]{%
      {\raggedright \textgraybox{\judgboxfontsize{\ensuremath{#1}}}\!\begin{tabular}[c]{l} #2 \end{tabular} %
      \!\!\textgraybox{\fontsz{8pt}{#3}}%
      \ifnum\OPTIONConf=1 \\[0pt] \else \\[6pt] \fi%
}}

\newcommand{\Rarrelim}{\ensuremath{{\arr}\text{E}}\xspace}

\newcommand{\Rarrintro}{\ensuremath{{\arr}\text{I}}\xspace}

\newcommand{\Runit}{\ensuremath{\Unit \text{I}}\xspace}

\newcommand{\Rsectintro}{\ensuremath{\land \text{I}}\xspace}

\newcommand{\Rsectelim}[1]{\ensuremath{\land \text{E}_{#1}}\xspace}
\newcommand{\Runivintro}{\ensuremath{\UnivSym \text{I}}\xspace}
\newcommand{\Runivelim}{\ensuremath{\UnivSym \text{E}}\xspace}

\newcommand{\Rvar}{\ensuremath{\mathsf{var}}\xspace}

\newcommand{\Rsub}{\ensuremath{\mathsf{sub}}\xspace}

\newcommand{\Ranno}{\ensuremath{\mathsf{right\text{-}anno}}\xspace}
\newcommand{\Rguard}{\ensuremath{\mathsf{left\text{-}anno}}\xspace}

\newcommand{\MKSUB}[1]{\ensuremath{{#1}{\subtype}}\xspace}
\newcommand{\subRefl}{\MKSUB{\text{refl}}}
\newcommand{\subArr}{\MKSUB{\arr}}

\newcommand{\subSectL}[1]{\MKSUB{{\sectty}\text{L}\ensuremath{_{#1}}}}
\newcommand{\subSectR}{\MKSUB{{\land}\text{R}}}

\newcommand{\subIndex}{\MKSUB{{i}\text{LR}}}
\newcommand{\subUnivL}{\MKSUB{{\UnivSym}\text{L}}}
\newcommand{\subUnivR}{\MKSUB{{\UnivSym}\text{R}}}

\newcommand{\cmtbegin}{\texttt{(*}}
\newcommand{\cmtend}{\texttt{*)}}
\newcommand{\annobegin}{\cmtbegin\texttt{\Lbrack}~\,}
\newcommand{\annoend}{\,~\texttt{\Rbrack}\cmtend}

\newcommand{\mydefaultarraystretch}{1.2}

\newcommand{\citepSequentCalculus}{\citep{Gentzen35}\xspace}

\newcommand{\dimIS}{\mathsf{dim}}

\newenvironment{displ}{\vspace{5pt} \begin{center} ~\!\!}{\end{center} \vspace{5pt}}

\newenvironment{mathdispl}{\ifnum\OPTIONConf=1\vspace{-10pt}\else\vspace{-15pt}\fi \begin{center}\begin{mathpar} ~\!\!}{\end{mathpar}\end{center}}

\def\quofile{15}
\newcommand{\startMquotes}{\immediate\openout \quofile=\jobname.quo}
\newcommand{\MquoteM}[4]{{\small \label{#3} \begin{quotation} #1 \vspace{-0.5em} \flushright{---#2} \end{quotation}} { \immediate\write\quofile{\expandafter\csname Mquoteentry\endcsname{}{#2}{#3}{\expandafter\csname #4\endcsname}}}}
\newcommand{\Mquote}[4]{{\small \label{#3} \begin{quotation} #1 \vspace{-0.5em} \flushright{---#2} \end{quotation}} { \immediate\write\quofile{\expandafter\csname Mquoteentry\endcsname{}{#2}{#3}{#4}}}}

\newcommand{\AllSym}{\forall}
\newcommand{\xAll}[1]{\AllSym#1}
\newcommand{\All}[1]{\xAll{#1}.\:}

\newcommand{\marginnoteRmlExamples}[1]{{}}

\newcommand{\term}[1]{\texttt{#1}}

\newcommand{\LBRACK}{\term{\Lbrack}}
\newcommand{\RBRACK}{\term{\Rbrack}}

\newcounter{zzInOinkComment}
\newcommand{\oinkkw}[1]{\ifnum\value{zzInOinkComment}=0{\usefont{T1}{cmtt}{m}{it}{\splatbold{#1}}}\else{\normalfont\textsl{#1}}\fi}

\newcommand{\lladdconj}{\mathrel{\binampersand}}

\newcommand{\RZZaddconjR}[1]{\ensuremath{{\lladdconj}\text{R}}\xspace}

\newdimen\zzlistingsize
\newdimen\zzlistingsizedefault
\zzlistingsize=\zzlistingsizedefault
\global\def\CommentCopter{0}
\newcommand{\Lstbasicstyle}{\fontsize{\zzlistingsize}{1.05\zzlistingsize}\ttfamily}
\newcommand{\keywordcopter}{\fontsize{1.0\zzlistingsize}{1.0\zzlistingsize}\bf}
\newcommand{\stupidcopter}{\if0\CommentCopter\keywordcopter\fi}
\newcommand{\commentcopter}{\def\CommentCopter{1}\fontsize{0.95\zzlistingsize}{1.0\zzlistingsize}\rmfamily\slshape}

\newlength{\zzlstwidth}
\newcommand{\setlistingsize}[1]{\zzlistingsize=#1%
\settowidth{\zzlstwidth}{{\Lstbasicstyle~}}}
\setlistingsize{\zzlistingsizedefault}

\definecolor{dHilite}{rgb}{0.9, 0.9, 0.6}
\definecolor{dRed}{rgb}{0.65, 0.0, 0.0}
\definecolor{DRED}{rgb}{0.65, 0.0, 0.0}
\definecolor{dGreen}{rgb}{0.0, 0.65, 0.0}
\definecolor{dDkGreen}{rgb}{0.0, 0.35, 0.0}
\definecolor{dBlue}{rgb}{0.0, 0.0, 0.65}
\definecolor{dPurple}{rgb}{0.65, 0.0, 0.65}
\definecolor{dDigPurple}{rgb}{0.5, 0.0, 0.5}
\definecolor{DDIGPURPLE}{rgb}{0.5, 0.0, 0.5}  %
\definecolor{dFaint}{rgb}{0.7, 0.7, 0.7}
\definecolor{dGray}{rgb}{0.5, 0.5, 0.5}
\definecolor{dDark}{rgb}{0.2, 0.2, 0.2}
\definecolor{dAlmostBlack}{rgb}{0.1, 0.1, 0.1}

\newcommand{\mathparleft}[1]{%
\end{mathpar} ~\\[-15pt]
      {\raggedright \begin{tabular}[c]{l} {#1} \end{tabular} %
      \\[-15pt]}%
\begin{mathpar}%
}

\newcommand{\explicitunivintrosym}{\Lambda}
\newcommand{\explicitunivintro}[2]{\explicitunivintrosym{#1 : #2}.\; }

\newcommand{\contextualtype}[2]{{#1}\LBRACK{#2}\RBRACK}
\newcommand{\contextualinst}[2]{{#1}\LBRACK{#2}\RBRACK}

\newcommand{\guardentails}{\mathrel{\texttt{>:>}}}
\newcommand{\Dec}{d}
\newcommand{\anno}[2]{\texttt{(}{#1} \,\texttt{:}\, {#2}\texttt{)}}
\newcommand{\lanno}[2]{{#1} \guardentails {#2}}
\newcommand{\xexistybind}[2]{\textkw{some}~#1 : #2}
\newcommand{\existybind}[2]{\xexistybind{#1}{#2}.\;}

\newcommand{\encode}[1]{\mathsf{trans}(#1)}

\newcommand{\shortandlongtitle}{Annotations for Intersection Typechecking}

\title{%
   \shortandlongtitle
}

\ifnum\OPTIONConf=1     %
\newcommand{\titlerunning}{\shortandlongtitle}
\newcommand{\authorrunning}{J.\ Dunfield}

\author{Joshua Dunfield
    \institute{Max Planck Institute for Software Systems
          \\ Kaiserslautern and Saarbr\"ucken, Germany}
    \email{j\hspace{-0.5pt}os~~\!\hspace{-3pt}u\hspace{-9.5pt}h\hspace{5pt}a@mpi-sws.org}
}
\else
\author{%
     Joshua Dunfield %
  \\
     \small Max Planck Institute for Software Systems
\\ \small Kaiserslautern and Saarbr\"ucken, Germany
}
\fi

\begin{document}
\maketitle

\begin{abstract}
  In functional programming languages, the classic form of annotation is
  a single type constraint on a term.  Intersection types add complications:
  a single term may have to be checked several times against different types,
  in different contexts, requiring annotation with several types. 
  Moreover, it is useful (in some systems, necessary) to indicate the context
  in which each such type is to be used.

  This paper explores the technical design space of annotations in systems
  with intersection types.  Earlier work~\citep{Dunfield04:Tridirectional} introduced
  \emph{contextual typing annotations}, which we now tease apart into
  more elementary mechanisms: a ``right hand'' annotation (the standard form),
  a ``left hand'' annotation (the context in which a right-hand annotation is to be used),
  a \emph{merge} that allows for multiple annotations, and an existential binder for
  index variables. 
  The most novel element is the left-hand annotation, which guards terms
  (and right-hand annotations) with a judgment that must follow from the
  current context.
\end{abstract}


\setcounter{footnote}{0}

\section{Introduction}

The origin of intersection lay in the analysis of the solvability of $\lambda$-terms;
the key early result was that, in a system with $\arr$ and $\sectty$, typeability and
strong normalization coincide~\citep{CDV81:IntersectionTypes}.
While pure type assignment is thus undecidable for intersection types,
systems that \emph{check} types of lightly-annotated programs, including systems
based on bidirectional typechecking, have had some success.  But
constructing a type-checking system from a type assignment system is not trivial.
A key issue is the design of the annotations.  The classic annotation form
$(e : A)$, which merely marks a term with a single type,
fails in intersection type systems that must check the same term several times,
in different contexts.  Furthermore, in systems with indexed types, we run into
problems with the scope of index variables; the simple mechanism of a term-level
binder fails, because intersections can be formed from types with different numbers
of quantifiers.

For guidance, we can look to logic and the form of hypothetical judgments: in
$\Gamma \entails \Delta$ we have, on the left, assumptions $\Gamma$ (implicitly conjoined,
because we wish to make several assumptions, each definite); on the right, we have conclusion
$\Delta$.  In the sequent calculus~\citepSequentCalculus, the conclusion is plural and implicitly
\emph{dis}joined: from a conjunction of assumptions, we conclude a disjunction of
conclusions.  This conforms to the internal duality of the sequent calculus.

The classic annotation form, $e : A$, seems to be ``on the right''.  It is
an obligation that constrains the type of $e$: ``I insist that 
$e$ have type $A$, and if you cannot satisfy this demand, typechecking should fail.''
(The term $e$ might have some other type $B$, but unless $B$ is a subtype of $A$
the demand is not met.  Also, in typecheckers that backtrack, like the
intersection-type checkers considered in this paper, the requirement that ``typechecking should fail''
means that the particular typing subproblem fails---the program could still typecheck.)
Writing $(e : A)$ does not correspond
to having an assumption $e : A$, because that would let us assume that $e$ has type $A$,
even if it should not have that type.  Further evidence in support
of right-handedness is that several systems with intersection types allow lists of
types in annotations, and these lists are
interpreted disjunctively, consistent with the sequent calculus where
lists of conclusions are interpreted disjunctively.

If the classic annotation $(e : A)$ is ``on the right'', what form of annotation is ``on the left''?
It is hard to imagine an annotation
that is not an obligation, or does not contribute to an obligation (leaving aside the sort
of annotation that is an explicit direction to ignore truth and charge ahead, as with
the \texttt{admit} of Coq~\citep{Coq:website}
or the \texttt{\%trustme} of Twelf~\citep{Twelf:website}).

We can, however, distinguish annotations that carry an obligation with respect to the
term on the right of the turnstile, such as $(e : A)$, from those that carry an
obligation with respect to the assumptions on the \emph{left} of the turnstile.
Writing such a ``left-hand'' annotation says, ``I insist on something about the
assumptions you have when you type this term, and if you cannot satisfy me, give up.''
Since the point of an assumption is to help conclude things, the
``something about the assumptions'' should be about what those assumptions entail.  The most
direct entailment is the use of a hypothesis: if $\Gamma = \{\Gamma_1, \dots, \Gamma_n\}$
then $\Gamma \entails \Gamma_k$ for $1 \leq k \leq n$, suggesting that we should
be able to write part of a context as a left-hand annotation.

The last piece of the puzzle is a way of writing more than one (right-hand) annotation.
It suffices to support a well-behaved special case of the unruly
\emph{merge construct}~\citep{Dunfield12:elaboration}.

\paragraph{Contents}  We start by giving an overview of annotations in intersection
type systems (\Sectionref{sec:overview}), then describe a language whose most notable
features are the left-hand \emph{guard annotation} (\Sectionref{sec:language}) and
a merge construct (\Sectionref{sec:opsem}).
Next, we extend that language with indexed types (\Sectionref{sec:indexed});
the presence of index variables leads us to another construct (an existential binder for
index variables).
In \Sectionref{sec:comparison},
we show that the features of the extended language---left- and right-hand annotations, plus the
merge construct and the existential binder---collectively subsume the
\emph{contextual typing annotations} developed in earlier work~\citep{Dunfield04:Tridirectional},
replacing one complicated construct with several simpler ones.  \Sectionref{sec:contextual-types}
compares our approach to contextual modal types.
Finally, we briefly discuss a prototype implementation (\Sectionref{sec:implementation})
and speculate on the usability of the approach (\Sectionref{sec:usability}).

\section{Overview} \Label{sec:overview}

For languages based on the ordinary $\lambda$-calculus,
the usual form of annotation is a single type, either around a term ($e : A$) or on
a bound variable ($\Lam{x:A} e$).  In such languages, the single type corresponds to typing:
exactly one subderivation types each subterm $e$.

In languages with intersection types, the introduction rule for intersection
replicates the same term in each premise:
\[
   \Infer{\Rsectintro}
        {\arrayenvb{\Dee_1 \\ \Gamma \entails e : A_1}
          \\
          \arrayenvb{\Dee_2 \\ \Gamma \entails e : A_2}}
        {\Gamma \entails e : A_1 \sectty A_2}
\]
Both $\Dee_1$ and $\Dee_2$ have as conclusion a typing for $e$; in general,
neither $A_k$ is a subtype of the other.  In general, we need both derivations because
the differences between $A_1$ and $A_2$ can lead to structural differences
in $\Dee_1$ and $\Dee_2$, and even in the contexts used inside
$\Dee_1$ and $\Dee_2$.

Assume a subtyping system in which the type $\Bits$ of bitstrings is refined
by $\Odd$ and $\Even$, denoting bitstrings of odd and even parity
(having an odd or even number of $1$s).  Appending a $1$ (written $x \cdot 1$)
should flip the parity, so
\[
      (\Lam{x} x \cdot 1) ~:~ (\Odd \arr \Even) \sectty (\Even \arr \Odd)
\]
In the typing derivation, we assume $x : \Odd$ inside the first branch of \Rsectintro
and $x : \Even$ inside the second:
\begin{mathdispl}
   \Infer{\Rsectintro}
      {
        \Infer{}
            {x : \Odd \entails x \cdot 1 : \Even}
            {\cdot \entails (\Lam{x} x \cdot 1) : (\Odd \arr \Even)}
        \\
        \Infer{}
            {x : \Even \entails x \cdot 1 : \Odd}
            {\cdot \entails (\Lam{x} x \cdot 1) : (\Even \arr \Odd)}
      }  
      {\cdot \entails (\Lam{x} x \cdot 1)
        ~:~
        (\Odd \arr \Even) \sectty (\Even \arr \Odd)}  
\end{mathdispl}
This function $\Lam{x} x \cdot 1$ is very simple; assuming the goal type
$(\Odd \arr \Even) \sectty (\Even \arr \Odd)$ is already known,
any reasonable typechecker should handle it without annotations
inside the function body.  But more complicated code might require internal
annotation.  Anyway, programmers should be able to write
unnecessary annotations if they want to.

Here, there is no single type we can write for the use of $x$ in
$x \cdot 1$: in the left side of the derivation, $x$ has type $\Odd$, and
in the right side, $x$ has type $\Even$.
To handle this issue, several systems with intersection types allow
\emph{lists} of types in annotations:
Forsythe~\citep{Reynolds88:Forsythe,Reynolds96:Forsythe}
and \citet[p.\ 21]{PierceThesis}
allow $\lambda$ arguments to be annotated with a sequence of types:
$\Lam{x : \Odd|\Even} x \cdot 1$;
the refinement typechecker SML-CIDRE~\citep{DaviesThesis}
allows terms to be annotated with lists of types, so we could write
$\Lam{x} (x : \Odd, \Even) \cdot 1$.

Intersection type inference is undecidable, but even intersection type \emph{checking}
is PSPACE-hard.  Unfortunately, unlike Hindley-Milner inference, which is intractable
in theory but polynomial in practice, intersection typechecking is expensive
in practice~\citep{Dunfield07:Stardust}.  A system
should, therefore, give the user a rich set of tools---such as annotations---to help make typechecking
practical.

Finally, in systems with indexed types and index-level variables, we need
to resolve a conflict between orderly variable scoping and intersection types.

Earlier work~\citep{Dunfield04:Tridirectional} described a \emph{contextual typing annotation}
that combined several features:

\begin{itemize}
\item \emph{contextuality}, guarding the type in the annotation with the context in which it makes sense;
\item \emph{multiplicity}, allowing more than one typing to be given, corresponding to different
  branches of intersection;
\item \emph{index variable linking}, maintaining index variable scoping even with
  intersection types.
\end{itemize}

We now recast the contextual typing annotation, separating it into constituent mechanisms
that collectively subsume it.
For contextuality, we introduce a \emph{guard} construct.
For multiplicity, we use a \emph{merge construct}~\citep{Dunfield12:elaboration}.
For index variable linking, we propose an existential binder.

\section{A Language with Guard Annotations}
\Label{sec:language}

\begin{figure}[htbp]
  \centering
  
  \begin{tabular}[t]{lr@{~~}c@{~~}ll}
        Types & $A, B, C$ & $\bnfas$ & $\Unit \bnfalt A \arr B \bnfalt A \sectty B$
\\[4pt]
        Terms & $e$ & $\bnfas$ & $x \bnfalt \unit \bnfalt \Lam{x}e \bnfalt e_1\,e_2$
\\ 
        & & & $\!\!\!\bnfalt \anno{e}{A}$ & standard (``right-hand'') annotation
\\
        & & & $\!\!\!\bnfalt \lanno{\Dec}{e}$ & guard (``left-hand'') annotation
\\[1pt]
        & & & $\!\!\!\bnfalt \Merge{e_1}{e_2}$ & merge
\\[4pt]
        Declarations & $\Dec$ & $\bnfas$ & $x : A$
\\[4pt]
        Contexts & $\Gamma$ & $\bnfas$ & $\cdot \bnfalt \Gamma, \Dec$
  \end{tabular}
  
  \caption{Types, terms, declarations and contexts}
  \FLabel{fig:syntax}
\end{figure}

We'll use a small functional language with intersection types,
a merge construct, and two kinds of annotations~(\Figureref{fig:syntax}).

\begin{figure}[t]
  \centering

  \begin{mathpar}
  \mathparleft{Subtyping}

  \vspace{-20pt} \\
     \Infer{\subRefl}
         {}
         {\Gamma \entails A \subtype A}
     \and
     \Infer{\subArr}
         {\Gamma \entails B_1 \subtype A_1
           \\
           \Gamma \entails A_2 \subtype B_2
           }
         {\Gamma \entails A_1 \arr A_2 \subtype B_1 \arr B_2
         }
     \and
     \Infer{\subSectL{k}}
         {\Gamma \entails  A_k \subtype B
         }
         {\Gamma \entails  A_1 \sectty A_2 \subtype B}
     \rulesep
     \Infer{\subSectR}
         {\Gamma \entails A \subtype B_1
           \\
           \Gamma \entails A \subtype B_2
         }
         { \Gamma \entails  A \subtype B_1 \sectty B_2
         }

  \mathparleft{Variables, $\Unit$, $\arr$}

  \vspace{-20pt} \\
    \Infer{\Rvar}
         { }
         {\Gamma_1, x : A, \Gamma_2 \entails x \has A}
    \and
    \Infer{\Runit}
         { }
         {\Gamma \entails \unit \against \Unit}
     \\
     \Infer{\Rarrintro}
         {\Gamma, x : A \entails e \against B}
         {\Gamma \entails \Lam{x} e \against A \arr B}
     \rulesep
     \Infer{\Rarrelim}
         {\Gamma \entails e_1 \has A \arr B
          \\
           \Gamma \entails e_2 \against A}
         {\Gamma \entails e_1\,e_2 \has B}
  \mathparleft{Intersection, subsumption, merge}
     \Infer{\Rsectintro}
         {\Gamma \entails e \against A_1
          \\
          \Gamma \entails e \against A_2}
         {\Gamma \entails e \against A_1 \sectty A_2}
     \rulesep
     \Infer{\Rsectelim{k}}
         {\Gamma \entails e \has A_1 \sectty A_2}
         {\Gamma \entails e \has A_k}
     \\
     \Infer{\Rsub}
         {\Gamma \entails e \has A
          \\
          \Gamma \entails A \subtype B}
         {\Gamma \entails e \against B}
     \and
     \Infer{\Rmerge{k}}
         {\Gamma \entails e_k \against A}
         {\Gamma \entails \Merge{e_1}{e_2} \against A}
     \rulesep
     \Infer{\Rmergeup{k}}
         {\Gamma \entails e_k \has A}
         {\Gamma \entails \Merge{e_1}{e_2} \has A}

  \mathparleft{Annotations}

  \vspace{-20pt} \\
     \Infer{\Ranno}
         {\Gamma \entails e \against A}
         {\Gamma \entails \anno{e}{A} \has A}
     \and
     \Infer{\Rguard{$\against$}}
         {\Gamma \entails x \against A
           \\
           \Gamma \entails e \against B}
         {\Gamma \entails \lanno{x : A}{e} \against B}
     \rulesep
     \Infer{\Rguard{$\has$}}
         {\Gamma \entails x \against A
           \\
           \Gamma \entails e \has B}
         {\Gamma \entails \lanno{x : A}{e} \has B}
  \end{mathpar}

  \vspace{-10pt}
  \caption{Subtyping and typing rules}
  \FLabel{fig:typing}
\end{figure}

\subsection{Bidirectional Typechecking}

Our type system is \emph{bidirectional}~\citep{Pierce00,Dunfield04:Tridirectional,Dunfield09}; see \citet{Dunfield09} for background.
This technique offers two major benefits over Damas-Milner type inference:
it works when annotation-free inference is undecidable, and
it produces more localized error messages.
Unlike constraint-based type inference, bidirectional typechecking does not inherently
require unification, nor the generation or manipulation of any constraints.
The basic idea of bidirectional typechecking is to separate checking of a
term against a known type from synthesis of an unknown type:
$\Gamma \entails e \against A$ means that $e$ checks against known
type $A$, while $\Gamma \entails e \has A$ means that $e$ synthesizes type $A$.
In the checking judgment, $\Gamma$, $e$ and $A$ are inputs to the typing algorithm.
In the synthesis judgment, $\Gamma$ and $e$ are inputs
and $A$ is output.  As usual, declarations of the form $x : A$ are added to $\Gamma$
through $\arr$-introduction (rule \Rarrintro); unlike in the Damas-Milner framework,
the type added is not a unification variable but a closed type.  In \Rarrintro, the type $A$
comes from the type $A \arr B$ that the $\lambda$-expression is checked against.

Bidirectional typechecking does need more type annotations than type inference.
However, by following the approach of~\citet{Dunfield04:Tridirectional}---%
checking introduction forms (like $\Lam{x} e$) and synthesizing the types of
elimination forms (like $e_1\,e_2$)---annotations are required only on redexes
like $(\Lam{x} e_1) e_2$ and recursive function declarations.  The need for annotations
is thus predictable; variations and refinements of this basic approach (such as trying
to synthesize the types of introduction forms) can further reduce the volume of annotations.

While we omit parametric polymorphism from this paper to focus on issues specific to intersection types,
it is straightforward to support parametric polymorphism, \emph{if} type abstraction
and application are explicit: given $e : \All{\alpha}B$, write $e[A]$ to instantiate $\alpha$ at $A$.
Such explicit instantiation is very inconvenient for the programmer.  It is possible, but not
entirely straightforward, to extend bidirectional typechecking with a form of existential
type variable~\citet{Dunfield09}.  This algorithm removes the need for explicit instantiation,
yet does not use unification, relying instead on a form of matching.

\subsection{Merging}

If either $e_1$ or $e_2$ has type $A$, then the merge $\Merge{e_1}{e_2}$ has type
$A$.  This construct first appeared in Forsythe~\citep{Reynolds96:Forsythe}.
Used in full generality~\citep{Dunfield12:elaboration}, the merge can encode a
variety of type system features, requires an elaboration-based semantics, and leads to
ambiguity if $e_1$ and $e_2$ have different operational behaviour.  In the
present setting, the purpose of the merge is just to let us annotate the same term in
different ways.  Used in this restricted fashion, erasing annotations from $e_1$ and $e_2$
yields the same term; thus, $e_1$ and $e_2$ have the same operational behaviour.
We discuss this point further in \Sectionref{sec:opsem}.

Since the merge is neither an introduction nor an elimination form, we can give
a synthesizing rule in addition to a checking rule; see \Figureref{fig:typing}.

Using a merge, the example $\Lam{x} x \cdot 1$ from the introduction
can be annotated as follows:
\[
  \Lam{x}\;
      \Merge
           {(x \cdot 1 : \Even)}
           {(x \cdot 1 : \Odd)}
\]
so it checks against $(\Odd \arr \Even) \sectty (\Even \arr \Odd)$.

\subsection{Guard Annotations}

Checking a function against intersection type leads to the function body being checked
several times against different return types, and even under varying typings of the
function's argument.  The latter motivates \emph{guards}.  A guard $\lanno{\Dec}{e}$
protects a term $e$ (say, the body of a function) with a declaration, so that the current typing context
$\Gamma$ must support the guarding
declaration $\Dec$.  For variable declarations $x : A$, this amounts to
$\Gamma \entails x \against A$.

We have both synthesis and checking typing rules for guards, ensuring that guards
can be placed anywhere the user chooses.

Using guards, we can annotate the example $\Lam{x} x \cdot 1$ so that
the choice of branch is fully determined:
\[
  \Lam{x}\;
      \Merge
           {\big(\lanno{x : \Odd}(x \cdot 1 : \Even)\big)}
           {\big(\lanno{x : \Even}(x \cdot 1 : \Odd)\big)}
\]

\subsection{Free Annotation} \Label{sec:noindex-free-annotation}

Given a term $e$ that can be typed with the bidirectional rules---that is, a term that already
has enough annotations for the typechecker---the user can freely choose to put in more
annotations, either right-hand annotations or guards.  If different annotations are needed
in the subderivations of \Rsectintro, the user can duplicate the term with a merge.

\section{Operational Semantics of Annotations and Merges}
\Label{sec:opsem}

We are working with a bidirectional type system.  For such a system, the easiest way to translate
the usual notions of preservation and progress is to give an equivalent type assignment system.
That is, we want rules deriving $\Gamma \entails e : A$ such that:

\begin{enumerate}[(1)]
\item if $\Gamma \entails e \against A$, then $\Gamma \entails e : A$;
\item if $\Gamma \entails e \has A$, then $\Gamma \entails e : A$.
\end{enumerate}

To show \emph{equivalence}, we would also need to consider the other direction: given some
$\Gamma \entails e : A$, can we derive appropriate bidirectional judgments?  We need not
answer this question to describe the operational semantics;
see \citet{Dunfield04:Tridirectional} for one answer.

\subsection{Left- and Right-Hand Annotations}

For standard and guard annotations, we can give a small-step operational semantics, but we
have a choice of approaches.
The first approach---standard in typed functional languages---is to erase the annotations,
so that the operational semantics does not mention them at all.
In this approach, we define an erasure function $\erasetypes{e}$:
\begin{displ}
  \begin{tabular}[t]{rcl}
      $\erasetypes{x}$ & $=$ & $x$
\\  $\erasetypes{\unit}$ & $=$ & $\unit$
\\  $\erasetypes{\Lam{x} e}$ & $=$ & $\Lam{x} \erasetypes{e}$
\\  $\erasetypes{e_1\,e_2}$ & $=$ & $\erasetypes{e_1}\,\erasetypes{e_2}$
\\[2pt]
  $\erasetypes{\anno{e}{A}}$ & $=$ & $\erasetypes{e}$
\\  $\erasetypes{\lanno{\Dec}{e}}$ & $=$ & $\erasetypes{e}$
  \end{tabular}
\end{displ}
The typing rules for left- and right-hand annotations have premises typing the inner expression $e$,
so this erasure function clearly preserves types.
Since $\anno{e}{A}$ and $\lanno{\Dec}{e}$ get erased, they need no reduction rules.

The second approach is to extend the definition of values:
\begin{displ}
  \begin{tabular}[t]{rcl}
    $v$ &$\bnfas$& $x \bnfalt \Lam{x} e \bnfalt \anno{v}{A} \bnfalt \lanno{\Dec}{v}$
  \end{tabular}
\end{displ}
and give reduction rules that drop the annotations.
\vspace{-12pt}
\begin{mathdispl}
  \Infer{}
      {}
      {\anno{e}{A} ~\step~ e}
  \and
  \Infer{}
      {}
      {\lanno{\Dec}{e} ~\step~ e}
\end{mathdispl}

As just noted, the typing rules for these constructs have premises typing $e$, so we can readily
extend an existing proof of type preservation to handle these new reduction rules.
Moreover, progress is maintained: if $\anno{e}{A}$ is well-typed and not a value,
then $e$ is not a value, and we can use the induction hypothesis on the premise typing $e$
to show that $e$, which $\anno{e}{A}$ steps to, is well-typed.

\subsection{Merges}

We still have to deal with the merge construct.
In this paper, we are interested in merges only as an annotation mechanism;
merges $\Merge{e_1}{e_2}$ used for that purpose
must have similar branches $e_1$ and $e_2$.  That is,
$e_1$ and $e_2$ are differently-annotated versions of some unannotated
``parent'' $e$.  We can apply the first approach---erasing annotations before
evaluation---by extending the definition of erasure:
\begin{displ}
  \begin{tabular}[t]{rcll}
      $\erasetypes{\Merge{e_1}{e_2}}$ & $=$ & $\erasetypes{e_1}$
      & if $\erasetypes{e_1} = \erasetypes{e_2}$
  \end{tabular}
\end{displ}
If merges are indeed always used purely as an annotation mechanism, the side condition will always hold.

We can also try to apply the second approach of reducing annotations during evaluation,
with reduction rules
\vspace{-18pt}
\begin{mathdispl}
  \Infer{}
      {}
      {\Merge{e_1}{e_2} ~\step~ e_1}
  \and
  \Infer{}
      {}
      {\Merge{e_1}{e_2} ~\step~ e_2}
\end{mathdispl}
These reduction rules introduce nondeterminism.
If we continue to assume, however, that $e_1$ and $e_2$
are differently-annotated versions of the same term, this nondeterminism is harmless:
we will end up with the same value, no matter which rule we apply.
Thus, we could omit one of the preceding two reduction rules, removing
the nondeterminism.

\subsection{Merges Without Restriction}

Giving an operational semantics to arbitrary uses of merge, where $e_1$ and $e_2$
may be entirely different, is more involved.  \citet{Dunfield12:elaboration} gives such
a semantics in two parts.  The first part is a system of reduction rules, including the two
above, for which the usual notions of preservation and progress fail to hold.  The second
part is an elaboration (more involved than erasure) to target terms $M$, which are
evaluated by a completely standard operational semantics.
This elaboration translates intersections to products (and unions to sums); the elaborating
version of \Rsectintro generates a pair, and the elaborating versions of \Rsectelim{k}
generate projections.

The central result in that paper is that if $e$
elaborates to $M$, evaluating the target term $M$ produces a value $W$ such that
\emph{there exists} some sequence of reductions of $e$ that yield an equivalent value $v$---one such that
$v$ elaborates to $W$.

\section{Extension to Indexed Types with Index Variables}
\Label{sec:indexed}

The above constructs collectively yield annotations that work when terms are
checked repeatedly under different contexts.  But this does not quite subsume contextual typing
annotations~\citep{Dunfield04:Tridirectional}, which were designed in the setting of
a system with indexed types as well as intersection (and union) types, and treat index-level
variables $a$, $b$ differently from term-level variables ($x$, $y$, etc.).

After setting the stage with some background on indexed types, we look at two alternatives
in language design and show how our approach works for both; for one of the alternatives,
one more language construct is needed.

\subsection{Indexed Types}

The kind of indexed types we consider here is exemplified by DML
\citep{Xi99popl,XiThesis}, and some of its descendants
\citep{Dunfield03:IntersectionsUnionsCBV,Dunfield04:Tridirectional,DunfieldThesis},
which added several features, most notably intersection and union types.
In these systems, users can index datatypes with \emph{index expressions} from a constraint domain
with decidable equality (at least).  The canonical example of such a domain is linear inequalities over integers;
dimensions (metres, seconds, etc.) form another useful domain~\citep{Dunfield07:Stardust}.

In contrast to dependent types, indices do not appear in terms $e$ (except within annotations)
and disappear completely during compilation; terms $e$ can never appear in indices.
Indexed type systems are parametric in the index domain.

We mostly follow (\Figureref{fig:indexed}) the notation of
\citet{Dunfield04:Tridirectional}.  Index expressions $i$ have
index sorts $\gamma$ (e.g.\ $\intIS$ or $\dimIS$); $a$ and $b$ are index-level
variables standing for index expressions; $P$ stands for propositions
over index expressions, such as equality $\indexeq$.
Types are extended with indexed datatypes
$\tau(i)$ (where $\tau$ is some inductive datatype $\tyname{list}$, $\tyname{tree}$, etc.)
and universal quantification over index variables.
(The use of $\UnivSym$ is traditional and, to readers
used to dependent types, has the advantage of suggesting the appropriate quantifier,
with the disadvantage of being easily confused with a genuine dependent $\Pi$.)
In practice, we also need existential quantification $\Exis{a}{\gamma} A$,
which we omit since it has no effect on the techniques described in this paper.

We assume that the constraint domain defines when two kinds of judgments
are derivable: $\Gamma \entails P$ (index assumptions in $\Gamma$
entail index proposition $P$) and $\Gamma \entails i : \gamma$ (index expression $i$,
which might include index variables declared in $\Gamma$, has index sort $\gamma$).
The only mandatory syntax in an index domain is $\indexeq$,
which is needed for subtyping.
In practice, the index expressions $i$ might include literal integers and operations like $i + i$;
the index propositions would include comparisons like $i < i$.

Practical bidirectional typechecking with indexed types, unlike bidirectional typing for the language
in previous sections of this paper, does involve constraints.  However, these constraints are just
over index expressions, not types, so the basic structure of the bidirectional approach need not change.
For a discussion of the techniques involved, see \citet{XiThesis} and \citet{DunfieldThesis}.

\begin{figure}[t]
  \centering
  
    \begin{tabular}[t]{lr@{~~}c@{~~}ll}
        Index variables & $a, b$ &
\\[2pt]
        Index sorts & $\gamma$ & $\bnfas$ & $\intIS \bnfalt \cdots$
\\[2pt]
        Index expressions & $i$ & $\bnfas$ & $a \bnfalt \cdots$
\\[2pt]
        Index propositions & $P$ & $\bnfas$ & $i \indexeq i \bnfalt \cdots$
\\[4pt]
        Types & $A, B, C$ & $\bnfas$ & $\cdots \bnfalt \tau(i) \bnfalt \Univ{a}{\gamma} A$
\\[4pt]
        Declarations & $\Dec$ & $\bnfas$ & $\cdots \bnfalt a : \gamma$ %
  \end{tabular}

  \caption{Indexed types}
  \FLabel{fig:indexed}
\end{figure}

\subsection{Indexed Types Without Binders}

The most syntactically economical formulation of indexed types does not
extend the term syntax at all (apart from the extension of the type language,
which changes the syntax of annotations).  Its subtyping and typing rules are
shown in \Figureref{fig:indexed-without-binders}.  Implicitly,
we assume that \subUnivR and \Runivintro rename the variable introduced
into the context if it already occurs in $\Gamma$.

\begin{figure}[t]
  \centering
  
  \begin{mathpar}
    \Infer{\subIndex}
        {\Gamma \entails i_1 \indexeq i_2}
        {\Gamma \entails \tau(i_1) \subtype \tau(i_2)}
    \and
    \Infer{\subUnivL}
        {\Gamma \entails i : \gamma
          \\
          \Gamma \entails [i/a]A \subtype B}
        {\Gamma \entails \Univ{a}{\gamma} A \subtype B}
    \rulesep
    \Infer{\subUnivR}
        {\Gamma, b : \gamma \entails A \subtype B}
        {\Gamma \entails A \subtype \Univ{b}{\gamma} B}
     \\
     \Infer{\Runivintro}
         {\Gamma, a : \gamma \entails e \against A}
         {\Gamma \entails e \against \Univ{a}{\gamma} A}
     \rulesep
     \Infer{\Runivelim}
         {\Gamma \entails e \has \Univ{a}{\gamma} A
         \\
         \Gamma \entails i : \gamma}
         {\Gamma \entails e \against [i/a] A}
  \end{mathpar}
  
  \vspace{-10pt}
  \caption{Subtyping and typing for indexed types (without term-level binders)}
  \FLabel{fig:indexed-without-binders}
\end{figure}

\begin{figure}[t]
  \centering

  \begin{mathpar}
    \Infer{}
        {\Gamma \entails i : \gamma
          \\
          \Gamma \entails [i/b]e \against A}
        {\Gamma \entails \existybind{b}{\gamma} e \against A}
    \rulesep
    \Infer{}
        {\Gamma \entails i : \gamma
          \\
          \Gamma \entails [i/b]e \has A}
        {\Gamma \entails \existybind{b}{\gamma} e \has A}
  \end{mathpar}

  \vspace{-10pt}
  \caption{Typing for the \textkw{some} binder}
  \FLabel{fig:some}
\end{figure}

Is that the end of the story?  No.  We have actually introduced a serious problem:  What does
it mean to mention an index variable $a$ in an annotation when there are no term-level binders?
The only thing that binds $a$ is $\UnivSym$, and the scope of the
binder $\Univ{a}{\gamma} A$ is just $A$.  And what if the implicit condition in \subUnivR
and \Runivintro is triggered and we have to rename the variable?  The user would be unable
to refer to the variable in annotations.

One way to solve this is to introduce an odd sort of binding construct,
$\existybind{a'}{\gamma} e$, which binds its variable $a'$ to some unwritten index
expression---one chosen by the typechecker to make everything work out.
An example:
\[
  \big(\Lam{x} \dots (\existybind{b}{\gamma} \lanno{x : \List(b*2)}{e}   ) \dots\big)
  ~\against~
  \Univ{a}{\intIS} \List(a*2) \arr \List(a)
\]
Within the inner term $e$, we can write (right-hand) annotations that mention $b$:
the typechecker chooses $b$ to be $a$, which satisfies the guard condition
$x \against \List(b*2)$.

The typing rules in \Figureref{fig:some} substitute
an index $i$ for $b$ in $e$, where $i$ is well-sorted in the actual context
$\Gamma$.  Thus, all annotations that mention $b$ will be renamed
so they make sense under $\Gamma$.
These rules do not require $i$ to be a variable: the following code is acceptable,
choosing $i$ to be $a * 2$.
\[
  \big(\Lam{x} \dots (\existybind{b}{\gamma} \lanno{x : \List(b)}{e}   ) \dots\big)
  ~\against~
  \Univ{a}{\intIS} \List(a*2) \arr \List(a)
\]
Non-renaming substitutions achieve a measure of robustness: the type being checked
against can, in some circumstances, change without requiring changes to internal
annotations.

\subsection{Indexed Types With Binders}

Alternatively, we can have an explicit term-level introduction form for $\Univ{a}{\gamma} A$:
  \begin{mathpar}
    \Infer{{\Runivintro}\text{-explicit}}
        {\Gamma, b : \gamma \entails e \against A}
        {\Gamma \entails \explicitunivintro{b}{\gamma}{e} \against A}
  \end{mathpar}
\citet{Dunfield04:Tridirectional} did not take this route, because
typing would fail for intersections of differently-quantified types.  For example,
the first conjunct of $(\Univ{a}{\gamma} A \arr A') \sectty (B \arr B')$
can type a term if it has a binder (for $a$), but the second conjunct cannot type a
term with a binder (since $B \arr B'$ has no $\UnivSym$).  With our merge construct,
we can write the term twice, with and without a binder.

\subsection{Free Annotation Revisited} \Label{sec:index-free-annotation}

Whether we have \textkw{some} binders or $\explicitunivintrosym$ binders,
we maintain the property mentioned in \Sectionref{sec:noindex-free-annotation}:
the user can always add an extra annotation if desired.

\begin{itemize}
\item  If we have \textkw{some}, the user will need to add a \textkw{some} binder
  for any index variable mentioned in annotations (left- and right-hand).
\item  If we have $\explicitunivintrosym$ and rule {\Runivintro}\text{-explicit}
  instead of \Runivintro, the user must already have put in the $\explicitunivintrosym$
  forms, and can refer to those bound index variables in annotations.
\end{itemize}

\vspace{-10pt}

\section{Comparison to Contextual Typing Annotations}
\Label{sec:comparison}

\begin{figure}[t]
  \centering
  
  \begin{mathpar}
        \Infer{\RctxSempty}
            {}
            {(\cdot \entails A) \ctxsubtype (\Gamma \entails A)}
        \and
        \Infer{\RctxSivar}
             { \Gamma \entails i : \gamma_0
               \\
               ([i/a]\,\Gamma_0 \entails [i/a]\,A_0) \ctxsubtype (\Gamma \entails A) }
             {   (a{:}\gamma_0, \Gamma_0 \entails A_0) \ctxsubtype (\Gamma \entails A)  }
         \and
         \Infer{\RctxSpvar}
             {  \Gamma \entails \Gamma(x) \subtype B_0
               \\
               (\Gamma_0 \entails A_0) \ctxsubtype (\Gamma \entails A)}
             { (x{:}B_0, \Gamma_0 \entails A_0) \ctxsubtype (\Gamma \entails A)}
          \vspace{6pt}
          \\
          \Infer{\Rctxanno}
              {(\Gamma_0 \entails A_0) \ctxsubtype (\Gamma \entails A)
                \\
                \Gamma \entails e \against A}
              {\Gamma \entails (e : \dots, (\Gamma_0 \entails A_0), \dots) \has A}    
  \end{mathpar}

  \vspace{-10pt}
  \caption{Rules for contextual typing annotations}
  \FLabel{fig:ctxanno}
\end{figure}

\begin{figure}[t]
  \centering
  
  \begin{array}[t]{rclll}
      \encode{x} &=& x
  \\
      \encode{\unit} &=& \unit
  \\
      \encode{\Lam{x} e} &=& \Lam{x} \encode{e}
  \\
      \encode{e_1\,e_2} &=& \encode{e_1}\;\encode{e_2}
  \\
      \encode{e : (\Gamma_1{\entails}A_1), \dots, (\Gamma_n{\entails}A_n)}
            &=&
            \Merge{\Merge{\encode{\Gamma_1 \entails A_1}}{\dots}}{\encode{\Gamma_n \entails A_n}}
    \\ 
      && \text{where $\encode{\Dec_1, \dots, \Dec_n \entails A}
        =  \lanno{\Dec_1}{\dots \lanno{\Dec_n}{\anno{\encode{e}}{A}}}$}
  \end{array}

  \vspace{-3pt}
 
  \caption{Translating contextual typing annotations} %
  \FLabel{fig:translate}
\end{figure}

We briefly review contextual typing annotations, introduced by
\citet{Dunfield04:Tridirectional}.  Such an annotation has a list $As$ of
typings $(\Gamma_1 \entails A_1, \dots, \Gamma_n \entails A_n)$.
The typing rule \Rctxanno (\Figureref{fig:ctxanno})
chooses a typing $\Gamma_0 \entails A_0$ and then uses a
\emph{contextual subtyping relation}
$(\Gamma_0 \entails A_0) \ctxsubtype (\Gamma \entails A)$,
which is derivable when $\Gamma$ is at least as strong as $\Gamma_0$,
that is, when $\Gamma$ satisfies all assumptions listed in $\Gamma_0$.
Declarations in $\Gamma_0$ thus should correspond to a sequence of guard
annotations.  Declarations of index variables in $\Gamma_0$, however,
are treated differently: the rule \RctxSivar behaves like the typing rules
for the \textkw{some} binder (\Figureref{fig:some}), effectively binding
variables declared in $\Gamma_0$ so they can be used in $A_0$.

In hindsight, contextual typing annotations combine all the mechanisms
in this paper---guard annotations, standard annotations, and merges:
program variable declarations $x : A$ in $\Gamma_0$ correspond to
a sequence of guard annotations, the type $A_0$ corresponds to a
standard annotation, and the multiplicity of typings corresponds to
merges.  Translating contextual typing annotations (\Figureref{fig:translate})
preserves typing:

\medskip

\begin{theorem}[Encoding Contextual Typing Annotations] ~\\
  If $\Gamma \entails e \against A$ (resp.\ $\has$) with rule \Rctxanno available
  then
  $\Gamma \entails  \encode{e'} \against A$ (resp.\ $\has$)
  without applying rule \Rctxanno.
\end{theorem}
\begin{proof}  By induction on the derivation.  All cases are
  straightforward except when \Rctxanno concludes the derivation.

  In that case, apply the i.h.\ resulting in $\encode{e_0'}$.
  This application of \Rctxanno uses one of the contextual typings, say
  $(\Gamma_k \entails A_k)$ where $\Gamma_k = \Dec_1, \dots, \Dec_n$;
  the $k$th branch of the merge created
  by $\encode{-}$ is $\lanno{\Dec_1}{\dots \lanno{\Dec_m}{\anno{\encode{e_0'}}{A}}}$.

  By rule \Ranno, $\Gamma \entails \encode{e_0'}{A} \has A$.
  
  By $m$ applications of \Rguard,
  \[
     \Gamma \entails~ \lanno{\Dec_1}{\dots \lanno{\Dec_m}{\anno{\encode{e_0'}}{A}}}
     ~~\has~ A
  \]
  Finally, apply \Rmergeupx as needed to pick out the
  $k$th branch of the merge created by $\encode{-}$.
\end{proof}

Given that we subsume contextual typing annotations, which approach
should be preferred when designing a language?  It is hard to give a universal
answer.  Generally speaking, simpler constructs are better than complicated
ones, but fewer constructs are better than many.  By the former criterion, the
mechanisms proposed in this paper win; by the latter, contextual typing
annotations win.  The particular design setting matters: if we need
some of these mechanisms already, their marginal cost is reduced.  This was the
case in the work that directly inspired this paper, elaboration-based typing of
intersections and unions~\citep{Dunfield12:elaboration}, where the merge
construct was already present.

\section{Comparison to Contextual Types}
\Label{sec:contextual-types}

There are several approaches to typing open code.  In one such approach,
contextual modal type theory~\citep{Nanevski08:CMTT}, the contextual type $\contextualtype{A}{\Psi}$
represents data of type $A$ closed under a context $\Psi$.  Providing a substitution for the
variables in $\Psi$ allows a term of type $\contextualtype{A}{\Psi}$
to yield a term of type $\contextualtype{A}{\cdot}$, closed
under the empty context---that is, a closed term.

Contextual types appear to subsume both guard annotations and our use of merges.  For example, instead of the
guard annotations in
$
  \Lam{x}\;
      \Merge
           {\big(\lanno{x : \Odd}(x \cdot 1 : \Even)\big)}
           {\big(\lanno{x : \Even}(x \cdot 1 : \Odd)\big)}
$
we could write
\[
\begin{array}[t]{l}
    \Lam{x}\;
        \Let{r}{
        (y \cdot 1)
        :
        \contextualtype{\Even}{y:\Odd}
        \sectty
        \contextualtype{\Odd}{y:\Even}
        }
        {}  
   \\
   \hspace{30pt}
          \contextualinst{r}{x/y}
\end{array}
\]
Checking $(y \cdot 1)$ against the first conjunct of the (ordinary right-hand)
annotation, $\contextualtype{\Even}{y:\Odd}$, shows that $(y \cdot 1)$ has
type $\Even$ when $y$ is substituted with a value of type $\Odd$.
The second conjunct is symmetric.  In the body of the \textkw{let}, we plug in $x$.
When we check the whole function against $(\Odd \arr \Even) \sectty (\Even \arr \Odd)$,
the variable $x$ will have type $\Odd$ in one subderivation of \Rsectintro,
and type $\Even$ in the other.  In each subderivation, using intersection elimination
gives $r$ a contextual type that can be eliminated by substituting $x$ for $y$.

Contextual types are versatile.  For example, they enable us to lift the binding of 
$r$ outside the function, and instantiate $r$ with different concrete contexts
(different substitutions for $y:\Even$) at several program points.
Extending typecheckers and compilers with such types, however,
is nontrivial~\citep{Pientka08:POPL}.  Introducing contextual types just to
support type annotations seems extravagant.  If contextual types are
already available in a language, of course, it could make sense to encode
the annotation mechanisms of this paper as contextual types, or
for programmers to write contextual types directly.

\section{Implementation}  \Label{sec:implementation}

Several of the ideas described above have been implemented in Stardust,
a typechecker (and compiler) for a small language in the Standard ML tradition.
In addition to intersection types and indexed types, Stardust supports union types,
datasort refinements and parametric polymorphism.

The implementation, with some examples, can be downloaded from
\url{http://stardust.qc.com}.  The syntax diverges slightly from the above presentation:

\begin{itemize}
\item the left-hand annotation $\lanno{d}{e}$ is written $\textkw{where}~d~\textkw{do}~e$;
\item type annotations can be given separately from their bindings; these annotations
  are similar to contextual type annotations, but with the $\lanno{d}{e}$ syntax for variable typings;
\item the \textkw{some} binder is (presently) only implemented for separate type annotations on
  bindings, not as an ordinary expression form.
\end{itemize}

An early version of Stardust was described in \citet{Dunfield07:Stardust,DunfieldThesis},
but the current version adds a number of important features, incorporating ideas
from \citet{Dunfield09,Dunfield12:elaboration}.

\section{Usability}  \Label{sec:usability}

We briefly consider some practical issues around the usability of our annotation
mechanisms.

The approach to bidirectional typechecking developed in \citet{Dunfield04:Tridirectional}
guarantees that right-hand annotations are needed only at redexes (most commonly,
recursive function declarations).  Once the user decides to add an annotation (whether
strictly required for typechecking, or for the purpose of documentation), the next step---of
adding a merge with left-hand annotations (or perhaps a contextual typing annotation)---is
fully determined: if the term needs to have different types under different contexts,
the user must add a merge and left-hand annotations.

The overall size of the annotations is hard to characterize.  Some examples of annotated
programs can be found in \citet{XiThesis} for bidirectional typechecking with indexed types
(but without intersections or contextual typings), \citet{DaviesThesis} for bidirectional typechecking
of refinement types and intersection types, and \citet{DunfieldThesis} for bidirectional typechecking
of refinement types, indexed types, intersection and union types.  Our experience with our
implementation is that for nontrivial uses of intersection and union types,
the performance of typechecking becomes highly problematic long before the
annotations become unacceptably long.  It is difficult to see how truly complex annotations could
be substantially reduced: if the annotations are complex, it is probably because the program specification
is nontrivial.

\subsection*{Acknowledgments}

I'd like to thank the ITRS reviewers for their comments and suggestions,
for both the pre- and post-proceedings versions; the workshop participants
for the discussion during the talk; Neelakantan R.\ Krishnaswami, for discussions and encouragement;
Frank Pfenning, for so many things.

{
\addtolength{\bibsep}{-1.0pt}
\bibliographystyle{plainnat}  %
\bibliography{intcomp}}

\end{document}